\documentclass{amsart}

\usepackage{amsmath}
\usepackage{amssymb}
\usepackage{tikz}
\usepackage{enumerate}
\usepackage{multicol}
\usepackage{graphicx}
\usepackage{hyperref}
\hypersetup{colorlinks=true,linkcolor=blue,citecolor=blue}

\newtheorem{theorem}{Theorem}[section]
\newtheorem{lemma}[theorem]{Lemma}

\theoremstyle{definition}
\newtheorem{definition}[theorem]{Definition}
\newtheorem{example}[theorem]{Example}
\newtheorem{corollary}[theorem]{Corollary}

\theoremstyle{remark}
\newtheorem{remark}[theorem]{Remark}

\numberwithin{equation}{section}

\author{Przemys{\l}aw Rola}

\address{Przemys{\l}aw Rola
\newline
\indent Institute of Mathematics
\newline
\indent Faculty of Mathematics and Computer Science
\newline
\indent Jagiellonian University
\newline
\indent ul. {\L}ojasiewicza 6
\newline
\indent 30-348 Krak\'ow, Poland}
\email{Przemyslaw.Rola@im.uj.edu.pl}


\begin{document}

\title[Arbitrage in markets with bid-ask spreads]{\large Arbitrage in markets with bid-ask spreads \\ \vspace{4 mm} The fundamental theorem of asset pricing in finite discrete time markets with bid-ask spreads and a money account}

\maketitle

\begin{abstract}
In this paper a finite discrete time market with an arbitrary state space and bid-ask spreads is considered. The notion of an equivalent bid-ask martingale measure (EBAMM) is introduced and the fundamental theorem of asset pricing is proved using (EBAMM) as an equivalent condition for no-arbitrage. The Cox-Ross-Rubinstein model with bid-ask spreads is presented as an application of our results.
\end{abstract}

\section{Introduction}
\label{intro}
The fundamental theorem of asset pricing, often called the Dalang-Morton-Willin\-ger theorem, states that for the standard discrete-time finite horizon model of security market there is no arbitrage if and only if the price process is a martingale with respect to an equivalent probability measure. However, the equivalent conditions for the absence of arbitrage in markets without friction were proposed and fully proved up to the $90$s the general problem of the equivalent conditions for the absence of arbitrage even in markets with proportional transaction costs is still open. On the other hand many great and significant work was done in this topic. We shortly recall some of papers devoted to multi-asset discrete-time models with friction.\\
\indent In paper \cite{KabRasStr} of Kabanov, R\'{a}sonyi, Stricker the equivalent conditions for the absence of so-called weak arbitrage opportunities (i.e. strict no-arbitrage) were given under the assumption of efficient friction. The general version of this theorem was proved by Kabanov and Stricker in \cite{KabStr2} but in the model with finite state space $\Omega$. Soon after Schachermayer in his famous paper \cite{Schach} gave the equivalent conditions for the absence of the so-called robust no-arbitrage. The general theorem states that the robust no-arbitrage is equivalent to the existence of a strictly consistent price system. Moreover, the robust no-arbitrage cannot be replaced be the strict no-arbitrage due to Schachermayer's counterexample presented in \cite{Schach}. Going further in this direction very interesting theorem was proved by Grigoriev in \cite{Grig}. The ge\-ne\-ral result of \cite{KabStr2} was extended for an arbitrary $\Omega$ in the special case of $2$ assets.\\
\indent One of the corollaries of that theorem states that in the market with bid and ask scalar processes $S^b$, $S^a$ and with a money account the absence of arbitrage is equivalent to the existence of a process $\tilde{S}$, which is a martingale under an equivalent probability measure and satisfies the inequalities $S^b \leq \tilde{S} \leq S^a$. In our ter\-mi\-no\-lo\-gy (which we introduce in a similar way as in \cite{JouKal}) it simply means that no arbitrage in the case of one risky asset and a money account is equivalent to the existence of a bid-ask consistent price system. For arbitrary $d$ risky assets this problem remains open. However, the author of \cite{Grig} suggested that its solution seems to be negative. One of purposes of our paper is to analyse the general model of market with bid-ask prices and a money account in order to research Grigoriev's question.\\
\indent Market model with bid and ask price processes was mainly developed in the famous paper of Jouini and Kallal \cite{JouKal} where the main result states that the absence of the so-called no free lunch is equivalent to the existence of a bid-ask consistent price system. It is noteworthy that the result of Grigoriev \cite{Grig} also strengthens the one of Jouini and Kallal in the case of one risky asset. In our paper we introduce the notion of an equivalent bid-ask martingale measure and prove that in the model with bid-ask spreads and an arbitrary state space $\Omega$ the existence of such a measure is equivalent to no arbitrage. On the other hand we show that the existence of (EBAMM) is equivalent to the existence of both supermartingale as well as submartingale consistent price systems under the same equivalent probability measure. We hope that the main theorem of our paper contributes to the solution of Grigoriev's hypothesis. In some sense it develops the results of \cite{JouKal} as well as \cite{Grig}. Moreover, the notion of (EBAMM) can be seen as a generalization of an equivalent martingale measure which is successfully used in markets without friction. It can also bypass the condition of the existence of a process $\tilde{S}$, which evolves between bid and ask price processes and is a martingale under some equivalent probability measure. It is rather obvious that such a process cannot exist in real and is only a useful tool for the pricing. We believe that this paper gives a contribution to the arbitrage theory in markets with friction and may reply in a comprehensible way to some of the questions.\\
\indent The paper is organized as follows. In section $2$ a model of a financial market and some basic definitions are introduced. In section $3$ the main theorem is presented and proved. The last chapter consists of examples as well as the applications of achieved results. Mainly the Cox-Ross-Rubinstein model with bid-ask spreads is presented.

\section{A mathematical model of a financial market}
\label{model}

Let $(\Omega, \mathcal{F}, \mathbb{P})$ be a complete probability space equipped with a discrete-time filtration $\mathbb{F} = (\mathcal{F}_t)_{t=0}^T$ such that $\mathcal{F}_T = \mathcal{F}$ and $T$ is a finite time horizon. Assume that in the market there are two processes $\underline{S} = (\underline{S}_t)_{t=0}^T = (\underline{S}_t^1, \ldots, \underline{S}_t^d)_{t=0}^T$ and $\overline{S} = (\overline{S}_t)_{t=0}^T = (\overline{S}_t^1, \ldots, \overline{S}_t^d)_{t=0}^T$, which are $d$-dimensional adapted to $\mathbb{F}$ and have strictly positive components, i.e. $\underline{S}_t^i > 0$ and $\overline{S}_t^i > 0$, $\mathbb{P}$-a.e. Furthermore we assume that $\underline{S}_t^i \leq \overline{S}_t^i$ for any $t=0,1,\ldots,T$ and $i=1,\ldots,d$. These processes model prices of shares for selling and buying respectively, i.e. at every moment $t$ the investor can buy or sell unlimited amounts of $i$-th shares at prices $\overline{S}_t^i$ and $\underline{S}_t^i$ respectively. We will call $\underline{S}$ the bid price process and $\overline{S}$ the ask price process. Analogously the pair $(\underline{S},\overline{S})$ will be called the bid-ask price process. Let us assume that there exists a money account or a bond in the market, which is a strictly positive predictable process $B = (B_t)_{t=0}^T$ and all transactions are calculated in units of this process. For simplicity we assume that $B_t \equiv 1$ for all $t=0, \ldots, T$ and to avoid technical ambiguity we put $\mathcal{F}_{-1} := \lbrace \emptyset, \Omega \rbrace$. It is noteworthy that this assumption do not restrict the generality of our model thanks to the discounting procedure. This procedure in details was described in chapter $2.1$ of \cite{DelSchach} for the case of markets without transaction costs.\\
\indent A trading strategy in the market is an $d$-dimensional process $H = (H_t)_{t=1}^T = (H_t^1, \ldots, H_t^d)_{t=1}^T$, which is predictable with respect to $\mathbb{F}$. We denote the set of all such strategies by $\mathcal{P}_T$. Let us also define its subsets $\mathcal{P}_T^+ := \lbrace H \in \mathcal{P}_T \; \vert \; H \geq 0 \rbrace$, $\mathcal{P}_T^- := \lbrace H \in \mathcal{P}_T \; \vert \; H \leq 0 \rbrace$ where the random vector $H \geq 0$ iff $H^i \geq 0$ for any $i=1,\ldots,d$. We use the notation $(H \cdot S)_t := \sum_{j=1}^t H_j \cdot \Delta S_j$ where $\cdot$ is the inner product in $\mathbb{R}^d$. Let $x = (x_t)_{t=1}^T$ be a value process in the market with bid-ask spreads for the strategy $H$ starting from $0$ units in bank and stock accounts, i.e. $x_t$ is defined as follows
$$x_t = x_t(H) := - \sum_{j=1}^t (\Delta H_j)^+ \cdot \overline{S}_{j-1} + \sum_{j=1}^t (\Delta H_j)^- \cdot \underline{S}_{j-1} + (H_t)^+ \cdot \underline{S}_t - (H_t)^- \cdot \overline{S}_t$$
where $\Delta H_j^i = H_j^i - H_{j-1}^i$ for any $i =1, \ldots, d$ and $j = 1, \ldots, t$. Especially we put $\Delta H_1^i = H_1^i$ and we will usually skip the symbol of the inner product. The random variable $x_t$ models the gain or loss occurred up to time $t$. The first sum is the aggregate purchase of assets up to time $t$ despite the second sum, which corresponds to the aggregate sales. Notice that at time $t$ we liquidate all positions in risky assets and the following equality is satisfied $\sum_{j=1}^t \Delta H_j = H_t$. It can be interpreted as follows. If we want to know the real value of our portfolio at time $t$ we should calculate it in units of a money account. To do this we should liquidate all positions in risky assets. Actually this procedure must not be carry out in real. We can use it only for calculating the value of our portfolio. In literature it is known as the immediate liquidation value of the portfolio (see e.g. \cite{CetJarProt} where on the other hand the marked-to-market value is considered).
\begin{remark}
Notice that all changes in units of assets must be obtained by borrowing or investing in a money account. Hence our position in the money account is uniquely determined by the strategy, which actually is self-financing.
\end{remark}
We will use the notation $L^0(\mathbb{R}^d,\mathcal{F}_t)$ for the set of $\mathcal{F}_t$-measurable random vectors taking values in $\mathbb{R}^d$ with the convention that $L^0(\mathbb{R}^d)$ stands for $L^0(\mathbb{R}^d,\mathcal{F}_T)$. In the case of random variables (i.e. $d=1$) we will simply use the abbreviations $L^0(\mathcal{F}_t) := L^0(\mathbb{R},\mathcal{F}_t)$ and $L^0 := L^0(\mathbb{R})$. Moreover let $L_+^0(\mathbb{R}^d,\mathcal{F}_t)$ denotes the subspace of $L^0(\mathbb{R}^d,\mathcal{F}_t)$ consisting of only non-negative random vectors. To simplify the notation we will use the same convention as previous, i.e. we will write $L_+^0(\mathbb{R}^d)$, $L_+^0(\mathcal{F}_t)$, $L_+^0$ in an appropriate situation. Furthermore the standard spaces $L^1$ and $L^{\infty}$ are treated in the same way.\\
\indent To make our reasoning much more clear we introduce for any $1 \leq t \leq t+k \leq T$ and $H \in L^0(\mathbb{R}^d,\mathcal{F}_{t-1})$ the following random variable:
$$x_{t-1,t+k}(H) := - (H)^+ \cdot \overline{S}_{t-1} + (H)^- \cdot \underline{S}_{t-1} + (H)^+ \cdot \underline{S}_{t+k} - (H)^- \cdot \overline{S}_{t+k}.$$
\indent Let us now define $\mathcal{R}_T := \lbrace x_T(H) \; \vert \; H \in \mathcal{P}_T \rbrace$ and the set of hedgeable claims, which is of the form
$$\mathcal{A}_T := \mathcal{R}_T - L_+^0.$$
By $\overline{\mathcal{A}}_T$ we denote the closure of $\mathcal{A}_T$ in probability. The following definition is crucial.

\begin{definition}
We say that there is no arbitrage in the market with bid-ask spreads if and only if
\begin{equation}
\mathcal{R}_T \cap L_+^0 = \lbrace 0 \rbrace. \tag{\textrm{NA}} \label{NA}
\end{equation}
\end{definition}
Notice that the condition (NA) is equivalent to $\mathcal{A}_T \cap L_+^0 = \lbrace 0 \rbrace$. Now we introduce another sets what simplify and make possible to prove the main theorem. Let us define for any $0 \leq j < t \leq T$:
\begin{equation}
\mathcal{R}_{j,t}^+ := \lbrace H \cdot (\underline{S}_{t} - \overline{S}_{j}) \; \vert \; H \in L_+^0(\mathbb{R}^d,\mathcal{F}_j) \rbrace,
\end{equation}
\begin{equation}
\mathcal{R}_{j,t}^- := \lbrace H \cdot (\overline{S}_{t} - \underline{S}_j) \; \vert -H \in L_+^0(\mathbb{R}^d,\mathcal{F}_j) \rbrace.
\end{equation}
Furthermore for any $1 \leq t \leq t+k \leq T$ we put
\begin{equation}
\mathbb{F}_{t-1,t+k} := \mathcal{R}_{t-1,t+k}^+ + \mathcal{R}_{t-1,t+k}^-, \qquad F_{t-1,t+k} := \mathbb{F}_{t-1,t+k} - L_+^0(\mathcal{F}_{t+k}).
\end{equation}
As counterparts of sets $\mathcal{R}_t$ and $\mathcal{A}_t$ we define for any $t=1,\ldots,T$ the following sets
\begin{equation} \label{F_sum_R}
\mathbb{F}_t := \sum\limits_{j<t} \mathcal{R}_{j,t}^+ + \sum\limits_{j<t} \mathcal{R}_{j,t}^- \qquad \text{and} \qquad F_t := \mathbb{F}_t - L_+^0(\mathcal{F}_t).
\end{equation}
Consequently we also introduce
\begin{equation}
\Lambda_T := \sum\limits_{t=1}^T \mathbb{F}_t - L_+^0.
\end{equation}
\begin{remark}
Notice that the sets $F_{t-1,t+k}$, $\mathbb{F}_{t-1,t+k}$, $\mathbb{F}_t$, $F_t$ and $\Lambda_t$ are convex cones and $\Lambda_T = \sum_{t=1}^T F_t$.
\end{remark}
\begin{lemma} \label{rem3}
Let (NA) $\mathcal{A}_T \cap L_+^0 = \lbrace 0 \rbrace$, then there is no arbitrage in the market with any time horizon $1 \leq t \leq T$, i.e. $\mathcal{A}_t \cap L_+^0(\mathcal{F}_t) = \lbrace 0 \rbrace$.
\end{lemma}
\begin{proof}
Notice that if $H$ is an arbitrage strategy in a model with the time horizon $t$ (so at time $t$ we liquidate all positions in stock) then it is also an arbitrage strategy in a model with a larger time horizon, especially with the time horizon $T$. It suffices to take the same strategy $H$ up to time $t$ and later $0$.
\end{proof}

We now introduce the definition of a consistent price system, similarly as it was done in \cite{JouKal}.

\begin{definition}
We say that a pair $(\tilde{S},\tilde{P})$ is a \emph{consistent price system} (CPS) in the market with bid-ask spreads when $\tilde{P}$ is a probability measure equivalent to $\mathbb{P}$ and $\tilde{S} = (\tilde{S}_t)_{t=0}^T$ is an $d$-dimensional process adapted to the filtration $\mathbb{F}$, which is a $\tilde{P}$-martingale and the following inequalities are satisfied
$$\underline{S}_t^i \leq \tilde{S}_t^i \leq \overline{S}_t^i, \quad \mathbb{P}\text{-a.e.}$$
for all $i =1, \ldots, d$ and $t = 0, \ldots, T$.\\
If the process $\tilde{S}$ is a $\tilde{P}$-supermartingale ($\tilde{P}$-submartingale) then we say that a pair $(\tilde{S},\tilde{P})$ is a \emph{supermartingale consistent price system} (supCPS) (a \emph{submartingale consistent price system} (subCPS) respectively).
\end{definition}

We introduce the notion of the so-called equivalent bid-ask martingale measure, which will play the similar role as an equivalent martingale measure in markets without friction.

\begin{definition}
We shall say that a probability measure $\mathbb{Q}$ is an \emph{equivalent bid-ask martingale measure} (EBAMM) for the bid-ask price process $(\underline{S},\overline{S})$ if $\mathbb{Q} \sim \mathbb{P}$, all $\overline{S}_t$ are integrable and the following inequalities are satisfied
\begin{equation}
\underline{S}_{t-1}^i \leq E_{\mathbb{Q}}(\overline{S}_t^i \vert \mathcal{F}_{t-1}) \quad \text{and} \quad E_{\mathbb{Q}}(\underline{S}_t^i \vert \mathcal{F}_{t-1}) \leq \overline{S}_{t-1}^i, \quad \mathbb{P}\text{-a.e.}
\end{equation}
for any $t = 1, \ldots, T$ and $i = 1, \ldots, d$.
\end{definition}

The interpretation of this measure is quite obvious. Let us consider (EBAMM) in the context of a stock market. If we buy shares at any time $t-1$ at price $\overline{S}_{t-1}^i$ we shouldn't expect, on average, that at time $t$ we sell shares at better price (i.e. at price $\underline{S}_t^i$) than we've bought them previous. On the other hand the analogous situation is if we short sale shares. The following lemma presents the straightforward relation between (CPS) and (EBAMM).

\begin{lemma} \label{CPS_EBAMM}
Assume that there exists (CPS) in the model. Then there exists an equivalent bid-ask martingale measure (EBAMM).
\end{lemma}
\begin{proof}
Let $(\tilde{S},\mathbb{Q})$ be a consistent price system. Then for any $t = 1, \ldots, T$ and $i = 1, \ldots, d$ we have the following inequalities:
$$E_{\mathbb{Q}}(\underline{S}_t^i \vert \mathcal{F}_{t-1}) \leq E_{\mathbb{Q}}(\tilde{S}_t^i \vert \mathcal{F}_{t-1}) = \tilde{S}_{t-1}^i \leq \overline{S}_{t-1}^i,$$
$$\underline{S}_{t-1}^i \leq \tilde{S}_{t-1}^i = E_{\mathbb{Q}}(\tilde{S}_t^i \vert \mathcal{F}_{t-1}) \leq E_{\mathbb{Q}}(\overline{S}_t^i \vert \mathcal{F}_{t-1}).$$
\end{proof}

\begin{remark}
The notion of (EBAMM) can be seen as a generalization of an e\-qui\-va\-lent martingale measure (EMM) in markets without friction. Indeed, when we assume that $\underline{S} = \overline{S}$, then our model comes down to a finite discrete time market model without transaction costs and (EBAMM) is actually the same as (EMM). Hence intuitively the notion suggests that if we could consider process $(\underline{S},\overline{S})$ as a whole then such a process should behave similarly to a martingale under an equivalent probability measure with precision to bid-ask spreads.
\end{remark}

\section{Main results}
\label{main results}

At the beginning of this chapter we present the sufficient condition for the absence of arbitrage, which is actually the existence of a consistent price system. This result is standard and well-known but we will prove it in our model using a slightly weaker assumption.

\begin{theorem}\label{th2}
Assume that there exists (supCPS) $(\hat{S},\mathbb{Q})$ and (subCPS) $(\check{S},\mathbb{Q})$. Let us define the set
\begin{equation}
\tilde{R}_T := \lbrace (\hat{H} \cdot \hat{S})_T + (\check{H} \cdot \check{S})_T \; \vert \; \hat{H} \in \mathcal{P}_T^+, \check{H} \in \mathcal{P}_T^- \rbrace.
\end{equation}
Then $\tilde{R}_T \cap L_+^0 = \lbrace 0 \rbrace$ and we have the absence of arbitrage in our model, i.e. $\mathcal{A}_T \cap L_+^0 = \lbrace 0 \rbrace$.
\end{theorem}

\begin{lemma} \label{arb_bounded}
Let us define
\begin{equation}
\tilde{R}_T^{b} := \lbrace (\hat{H} \cdot \hat{S})_T + (\check{H} \cdot \check{S})_T \; \vert \; \hat{H} \in \mathcal{P}_T^+, \check{H} \in \mathcal{P}_T^- \; \text{where} \; \hat{H}, \check{H} \; \text{are bounded} \rbrace.
\end{equation}
Then the condition $\tilde{R}_T \cap L_+^0 = \lbrace 0 \rbrace$ is equivalent to $\tilde{R}_T^{b} \cap L_+^0 = \lbrace 0 \rbrace$.
\end{lemma}
\begin{proof}
Notice that the condition $\tilde{R}_T \cap L_+^0 = \lbrace 0 \rbrace$ is equivalent to the absence of arbitrage for any one-step model, i.e. in our notation
\begin{equation}
\lbrace \hat{\eta} \Delta \hat{S}_t + \check{\eta} \Delta \check{S}_t \colon \hat{\eta}, - \check{\eta} \in L_+^0(\mathcal{F}_{t-1}) \rbrace \cap L_0^+(\mathcal{F}_t) = \lbrace 0 \rbrace
\end{equation}
for any $t = 1, \ldots, T$. We will use the analogous reasoning as in \cite{KabSaf}, see chapter $2.1.1$. Assume that we have the absence of arbitrage in any one-step model. We show that there is no arbitrage. Take the smallest $t \leq T$ such that $\tilde{R}_t \cap L_+^0(\mathcal{F}_t) \neq \lbrace 0 \rbrace$ and notice that $1 < t < T$. Hence there exist two strategies $\hat{H} \in \mathcal{P}_t^+$, $\check{H} \in \mathcal{P}_t^-$ such that
\begin{equation}
(\hat{H} \cdot \hat{S})_t + (\check{H} \cdot \check{S})_t \geq 0 \quad \text{and} \quad \mathbb{P}((\hat{H} \cdot \hat{S})_t + (\check{H} \cdot \check{S})_t > 0) > 0.
\end{equation}
Due to the choice of $t$, either the set $\Gamma^{'} := \lbrace (\hat{H} \cdot \hat{S})_{t-1} + (\check{H} \cdot \check{S})_{t-1} < 0 \rbrace$ is of strictly positive probability (we put $\hat{\eta} := 1{\hskip -4.0 pt}\hbox{1}_{\Gamma^{'}}\hat{H}_t$, $\check{\eta} := 1{\hskip -4.0 pt}\hbox{1}_{\Gamma^{'}}\check{H}_t$) or the set $\Gamma^{''} := \lbrace (\hat{H} \cdot \hat{S})_{t-1} + (\check{H} \cdot \check{S})_{t-1} = 0 \rbrace$ is of full measure (we take $\hat{\eta} := 1{\hskip -4.0 pt}\hbox{1}_{\Gamma^{''}}\hat{H}_t$, $\check{\eta} := 1{\hskip -4.0 pt}\hbox{1}_{\Gamma^{''}}\check{H}_t$). In any case we have a contradiction. Therefore we can assume that there exists $H_t \in L^0(\mathbb{R}^d,\mathcal{F}_{t-1})$ satisfying the following conditions
\begin{equation} \label{A1}
H_t \Delta \tilde{S}_t \geq 0, \; \mathbb{P}\text{-a.e.} \quad \text{and} \quad \mathbb{P}(H_t \Delta \tilde{S}_t > 0) > 0.
\end{equation}
It suffices to show that there exists $\tilde{H}_t \in L^0(\mathbb{R}^d,\mathcal{F}_{t-1})$, which is bounded and satisfies the condition (\ref{A1}).
One can take
\begin{displaymath}
\tilde{H}_t := \left\{ \begin{array}{ll}
\frac{H_t}{\Vert H_t \Vert} \quad H_t \neq 0,\\
0 \qquad \; H_t=0.
\end{array} \right.
\end{displaymath}
It is also possible to use the arguments from \cite{KabSaf} (chapter $2.1.1.$). Let us define the sequence $H_t^n := H_t 1{\hskip -4.0 pt}\hbox{1}_{\lbrace \Vert H_t \Vert \leq n \rbrace}$. Then there exists sufficiently large $n \in \mathbb{N}$ such that $H_t^n$ satisfies (\ref{A1}).
\end{proof}

\begin{proof}[Proof of Theorem \ref{th2}]
By Lemma \ref{arb_bounded} it suffices to prove that $\tilde{R}_T^{b} \cap L_+^0 = \lbrace 0 \rbrace$. Let $X = (\hat{H} \cdot \hat{S})_T + (\check{H} \cdot \check{S})_T \in \tilde{R}_T^{b} \cap L_+^0$. Hence $(\hat{H} \cdot \hat{S})_T + (\check{H} \cdot \check{S})_T \geq 0$ and in particular $H$ is a bounded strategy. We show that $E_{\mathbb{Q}}[(\hat{H} \cdot \hat{S})_T + (\check{H} \cdot \check{S})_T] \leq 0$. Using the assumption that $\hat{S}$ is a $\mathbb{Q}$-supermartingale and $\check{S}$ a $\mathbb{Q}$-submartingale we get $E_{\mathbb{Q}}(\hat{H}_t \Delta \hat{S}_t \vert \mathcal{F}_{t-1}) = \hat{H}_t E_{\mathbb{Q}}(\Delta \hat{S}_t \vert \mathcal{F}_{t-1}) \leq 0$. Analogously $E_{\mathbb{Q}}(\check{H}_t \Delta \check{S}_t \vert \mathcal{F}_{t-1}) \leq 0$. Summing up
\begin{equation}
E_{\mathbb{Q}}[(\hat{H} \cdot \hat{S})_T + (\check{H} \cdot \check{S})_T] \leq 0.
\end{equation}
Hence $X = 0$, $\mathbb{Q}$-a.e. and from the equivalence of measures $X = 0$, $\mathbb{P}$-a.e. We show now that $\mathcal{A}_T \cap L_+^0 = \lbrace 0 \rbrace$. Take any $\xi \in \mathcal{A}_T \cap L_+^0$. Then the following inequalities are satisfied:
$$0 \leq \xi \leq - \sum_{t=1}^T (\Delta H_t)^+ \overline{S}_{t-1} + \sum_{t=1}^T (\Delta H_t)^- \underline{S}_{t-1} + (H_T)^+ \underline{S}_T - (H_T)^- \overline{S}_T.$$
Let us notice that for any strategy $H \in \mathcal{P}_T$ there exist strategies $\hat{H} \in \mathcal{P}_T^+$ and $\check{H} \in \mathcal{P}_T^-$ such that $\Delta H_t^i = \Delta \hat{H}_t^i + \Delta \check{H}_t^i$. We can construct them as follows:
$$\text{if} \quad H_t^i \geq 0 \; \text{on the set} \; \lbrace H_{t-1}^i \geq 0 \rbrace \quad \text{then} \quad \Delta \hat{H}_t^i := \Delta H_t^i, \; \Delta \check{H}_t^i := 0,$$
$$\text{if} \quad H_t^i < 0 \; \text{on the set} \; \lbrace H_{t-1}^i < 0 \rbrace \quad \text{then} \quad \Delta \hat{H}_t^i := 0, \; \Delta \check{H}_t^i := \Delta H_t^i,$$
$$\text{if} \quad H_t^i \geq 0 \; \text{on the set} \; \lbrace H_{t-1}^i < 0 \rbrace \quad \text{then} \quad \Delta \hat{H}_t^i := H_t^i, \; \Delta \check{H}_t^i := - H_{t-1}^i,$$
$$\text{if} \quad H_t^i < 0 \; \text{on the set} \; \lbrace H_{t-1}^i \geq 0 \rbrace \quad \text{then} \quad \Delta \hat{H}_t^i := - H_{t-1}^i, \; \Delta \check{H}_t^i := H_t^i.$$
It means that we split the strategy into another two strategies, which consist of long and short positions only. Moreover, when $\Delta H_t \geq 0$ then $\Delta \hat{H}_t^i, \Delta \check{H}_t^i \geq 0$ and on the other hand if $\Delta H_t < 0$ then $\Delta \hat{H}_t^i, \Delta \check{H}_t^i < 0$. Hence we always have $(\Delta H_t^i)^+ = (\Delta \hat{H}_t^i)^+ + (\Delta \check{H}_t^i)^+$ as well as $(\Delta H_t^i)^- = (\Delta \hat{H}_t^i)^- + (\Delta \check{H}_t^i)^-$. Summing up, we get
$$\xi \leq - \sum_{t=1}^T (\Delta H_t)^+ \overline{S}_{t-1} + \sum_{t=1}^T (\Delta H_t)^- \underline{S}_{t-1} + (H_T)^+ \underline{S}_T - (H_T)^- \overline{S}_T =$$
$$= - \sum_{t=1}^T (\Delta \hat{H}_t)^+ \overline{S}_{t-1} + \sum_{t=1}^T (\Delta \hat{H}_t)^- \underline{S}_{t-1} + (\hat{H}_T)^+ \underline{S}_T - (\hat{H}_T)^- \overline{S}_T +$$
$$- \sum_{t=1}^T (\Delta \check{H}_t)^+ \overline{S}_{t-1} + \sum_{t=1}^T (\Delta \check{H}_t)^- \underline{S}_{t-1} + (\check{H}_T)^+ \underline{S}_T - (\check{H}_T)^- \overline{S}_T =: (\star)$$
Notice that the following inequalities are satisfied
$$\underline{S}_t^i \leq \hat{S}_t^i \leq \overline{S}_t^i \quad \text{and} \quad \underline{S}_t^i \leq \check{S}_t^i \leq \overline{S}_t^i, \quad \mathbb{P}\text{-a.e.}$$
for any $t=0, \ldots, T$ and $i=1, \ldots, d$. Hence we can write the next inequality, i.e.
$$(\star) \leq - \sum_{t=1}^T (\Delta \hat{H}_t)^+ \hat{S}_{t-1} + \sum_{t=1}^T (\Delta \hat{H}_t)^- \hat{S}_{t-1} + (\hat{H}_T)^+ \hat{S}_T - (\hat{H}_T)^- \hat{S}_T +$$
$$- \sum_{t=1}^T (\Delta \check{H}_t)^+ \check{S}_{t-1} + \sum_{t=1}^T (\Delta \check{H}_t)^- \check{S}_{t-1} + (\check{H}_T)^+ \check{S}_T - (\check{H}_T)^- \check{S}_T=$$
$$= - \sum_{t=1}^T\limits \Delta \hat{H}_t \hat{S}_{t-1} + \hat{H}_T \hat{S}_T - \sum_{t=1}^T\limits \Delta \check{H}_t \check{S}_{t-1} + \check{H}_T \check{S}_T = (\hat{H} \cdot \hat{S})_T + (\check{H} \cdot \check{S})_T.$$
Then
\begin{equation}
0 \leq \xi \leq (\hat{H} \cdot \hat{S})_T + (\check{H} \cdot \check{S})_T.
\end{equation}
Due to the condition $\tilde{R}_T \cap L_+^0 = \lbrace 0 \rbrace$ we get $(\hat{H} \cdot \hat{S})_T + (\check{H} \cdot \check{S})_T = 0$, $\mathbb{P}$-a.e. and hence $\xi = 0$, $\mathbb{P}$-a.e.
\end{proof}

\begin{remark}
Notice that if there exists (CPS) in the market then the assumptions of Theorem \ref{th2} are also satisfied and we have the absence of arbitrage in our model, i.e. $\mathcal{A}_T \cap L_+^0 = \lbrace 0 \rbrace$.
\end{remark}

Before we formulate the Fundamental theorem we present and prove some technical lemmas, which will play a role in our theory.

\begin{lemma} \label{lem_A_F}
For any $t=1, \ldots, T$ holds the inclusion $F_t \subset \mathcal{A}_t$.
\end{lemma}
\begin{proof}
Notice that it suffices to show that $\mathbb{F}_t \subset \mathcal{A}_t$ where $\mathbb{F}_t$ is defined as in (\ref{F_sum_R}). Take any $\Pi \in \mathbb{F}_t$. By definition we can assume that is of the form
$$\Pi = - \sum_{j=1}^t \theta_j \overline{S}_{j-1} + \sum_{j=1}^t \tilde{\theta}_j \underline{S}_{j-1} + \sum_{j=1}^t \theta_j \underline{S}_t - \sum_{j=1}^t \tilde{\theta}_j \overline{S}_t,$$
where $\Theta = (\theta_j)_{j=1}^t$, $\tilde{\Theta} = (\tilde{\theta}_j)_{j=1}^t$ are predictable and non-negative processes. Notice that there exist maybe another predictable and non-negative processes $\vartheta = (\vartheta_j)_{j=1}^t$, $\tilde{\vartheta} = (\tilde{\vartheta}_j)_{j=1}^t$ such that for any $j=1, \ldots, t$ we have $\lbrace \vartheta_j > 0, \; \tilde{\vartheta}_j > 0 \rbrace = \emptyset$, a.e. and the following inequality is satisfied
$$\Pi \leq \Xi := - \sum_{j=1}^t \vartheta_j \overline{S}_{j-1} + \sum_{j=1}^t \tilde{\vartheta}_j \underline{S}_{j-1} + \sum_{j=1}^t \vartheta_j \underline{S}_t - \sum_{j=1}^t \tilde{\vartheta}_j \overline{S}_t, \quad \mathbb{P}\text{-a.e.}$$
Let us define the strategy $H = (H_j)_{j=1}^t \in \mathcal{P}_t$ as follows
$$\Delta H_j := (\Delta H_j)^+ - (\Delta H_j)^- \; \text{where} \quad (\Delta H_j)^+ := \vartheta_j \quad \text{and} \quad (\Delta H_j)^- := \tilde{\vartheta}_j.$$
Moreover, we put $H_1 := \Delta H_1$ and $H_j := \Delta H_j + H_{j-1}$ for $j>1$. Notice that $H$ is a well defined strategy. Furthermore
$$\sum_{j=1}^t (\Delta H_j)^- \overline{S}_t - \sum_{j=1}^t (\Delta H_j)^+ \underline{S}_t + H_t^+ \underline{S}_t - H_t^- \overline{S}_t =$$
$$= (H_t^+ - \sum_{j=1}^t (\Delta H_j)^+) \underline{S}_t - (H_t^- - \sum_{j=1}^t (\Delta H_j)^-) \overline{S}_t$$
and the following equalities are satisfied
$$\sum_{j=1}^t (\Delta H_j)^+ - \sum_{j=1}^t (\Delta H_j)^- = \sum_{j=1}^t \Delta H_j = H_t = H_t^+ - H_t^-.$$
Therefore we have $H_t^+ - \sum\limits_{j=1}^t (\Delta H_j)^+ = H_t^- - \sum\limits_{j=1}^t (\Delta H_j)^-$ and $H_t^+ \leq \sum\limits_{j=1}^t (\Delta H_j)^+$. Let us define the random variable $r := (H_t^+ - \sum\limits_{j=1}^t (\Delta H_j)^+) \underline{S}_t - (H_t^- - \sum\limits_{j=1}^t (\Delta H_j)^-) \overline{S}_t$. By the previous observation $r = (H_t^+ - \sum\limits_{j=1}^t (\Delta H_j)^+)(\underline{S}_t - \overline{S}_t) \geq 0$ what simply means that $r \in L_+^0(\mathcal{F}_t)$. Hence
$$\Pi + r \leq \Xi +r = - \sum_{j=1}^t (\Delta H_j)^+ \overline{S}_{j-1} + \sum_{j=1}^t (\Delta H_j)^- \underline{S}_{j-1} + (H_t)^+ \underline{S}_t - (H_t)^- \overline{S}_t.$$
Obviously $\Xi +r = x_t(H) \in \mathcal{R}_t$ and $\Pi \leq x_t(H) - r$. Furthermore there exists a random variable $\tilde{r} \in L_+^0(\mathcal{F}_t)$ such that $\Pi = x_t(H) - r - \tilde{r}$. It suffices to define $\tilde{r} := \Xi - \Pi$. Hence we get that $\Pi \in \mathcal{A}_t$.
\end{proof}

\begin{remark}
It is not clear whether $\Lambda_T \subset \mathcal{A}_T$ or not. We only know that $\Lambda_T \subset \sum\limits_{t=1}^T \mathcal{A}_t$.
\end{remark}

\begin{remark} \label{rem_A_F}
Notice that for any $\Pi \in F_T$ (respectively $\mathbb{F}_T$) there exists a strategy $H \in \mathcal{P}_T$ and a random variable $r \in L_+^0$ such that $\Pi = x_T(H) - r$.
\end{remark}

\begin{lemma} \label{lem_Ftk}
For any $1 \leq t \leq t+k \leq T$ the following inclusions hold $F_{t-1,t+k} \subset F_{t+k} \subset \mathcal{A}_{t+k}$ and for any $x \in F_{t-1,t+k}$ there exists $H_t \in L^0(\mathbb{R}^d,\mathcal{F}_{t-1})$ and $r \in L_+^0(\mathcal{F}_{t+k})$ such that $x = x_{t-1,t+k}(H_t) - r$.
\end{lemma}
\begin{proof}
Fix $t$, $k$ such that $1 \leq t \leq t+k \leq T$ and take any $x \in F_{t-1,t+k}$. Let $x = \Pi - l$ where
$$\Pi = - \theta \overline{S}_{t-1} + \tilde{\theta} \underline{S}_{t-1} + \theta \underline{S}_{t+k} - \tilde{\theta} \overline{S}_{t+k}$$
and $\theta, \tilde{\theta} \in L_+^0(\mathbb{R}^d,\mathcal{F}_{t-1})$, $l \in L_+^0(\mathcal{F}_{t+k})$. Notice that there exist maybe another random vectors $\vartheta, \tilde{\vartheta} \in L_+^0(\mathbb{R}^d,\mathcal{F}_{t-1})$ such that $\lbrace \vartheta^i > 0, \; \tilde{\vartheta}^i > 0 \rbrace = \emptyset$, a.e. for any $i=1,\ldots,d$. Because at the same time we buy and short sale shares on these sets, hence we only pay transaction costs due to bid-ask spreads and the following inequality is satisfied
$$\Pi \leq \Xi := - \vartheta \overline{S}_{t-1} + \tilde{\vartheta} \underline{S}_{t-1} + \vartheta \underline{S}_{t+k} - \tilde{\vartheta} \overline{S}_{t+k} \in \mathcal{R}_{t+k}, \quad \mathbb{P}\text{-a.e.}$$
Now let us define the random vector $H_t := \vartheta - \tilde{\vartheta}$. Notice that $H_t \in L^0(\mathbb{R}^d,\mathcal{F}_{t-1})$ and $(H_t^i)^+ = \vartheta^i$, $(H_t^i)^- = \tilde{\vartheta}^i$. Moreover the random variable $\tilde{l} := \Xi - \Pi \in L_+^0(\mathcal{F}_{t+k})$ and we get the equality $x = \Pi - l = \Xi -l - \tilde{l}$. Let $r := l + \tilde{l} \in L_+^0(\mathcal{F}_{t+k})$. Then we have
$$x = - (H_t)^+ \cdot \overline{S}_{t-1} + (H_t)^- \cdot \underline{S}_{t-1} + (H_t)^+ \cdot \underline{S}_{t+k} - (H_t)^- \cdot \overline{S}_{t+k} - r \in F_{t-1,t+k}.$$
Also as we see $x \in F_{t+k}$ and by Lemma \ref{lem_A_F} $F_{t+k} \subset \mathcal{A}_{t+k}$.
\end{proof}

The following theorem presents the equivalent conditions for the absence of arbitrage in markets with bid-ask spreads and a money account.

\begin{theorem}[Fundamental theorem]\label{th1}
The following conditions are equivalent:\\
\indent (a) $\mathcal{A}_T \cap L_+^0 = \lbrace 0 \rbrace$ (NA);\\
\indent (b) $F_t \cap L_+^0(\mathcal{F}_t) = \lbrace 0 \rbrace$ for any $t=1, \ldots, T$;\\
\indent (c) $F_{t-1,t+k} \cap L_+^0(\mathcal{F}_{t+k}) = \lbrace 0 \rbrace$ for any $1 \leq t \leq t+k \leq T$;\\
\indent (d) $F_{t-1,t+k} \cap L_+^0(\mathcal{F}_{t+k}) = \lbrace 0 \rbrace$ and $F_{t-1,t+k} = \overline{F}_{t-1,t+k}$ for any\\
\indent \indent $1 \leq t \leq t+k \leq T$;\\
\indent (e) $\overline{F}_{t-1,t+k} \cap L_+^0(\mathcal{F}_{t+k}) = \lbrace 0 \rbrace$ for any $1 \leq t \leq t+k \leq T$;\\
\indent (f) there exists an equivalent bid-ask martingale measure $\mathbb{Q}$ for the bid-ask\\
\indent \indent process $(\underline{S}, \overline{S})$ such that $\frac{d\mathbb{Q}}{d\mathbb{P}} \in L^{\infty}$ (EBAMM);\\
\indent (g) there exists supCPS $(\hat{S},\mathbb{Q})$ and subCPS $(\check{S},\mathbb{Q})$ such that $\frac{d\mathbb{Q}}{d\mathbb{P}} \in L^{\infty}$.
\end{theorem}

In the proof of Theorem \ref{th1} the following results will be used. Their proofs can be found e.g. in \cite{KabStr}.

\begin{lemma} \label{lem1}
Let $X_n$ be a sequence of random vectors taking values in $\mathbb{R}^d$ such that for almost all $\omega \in \Omega$ we have $\liminf \Vert X_n(\omega) \Vert < \infty$. Then there is a sequence of random vectors $Y_n$ taking values in $\mathbb{R}^d$ satisfying the following conditions:\\
(1) $Y_n$ converges pointwise to $Y$  almost surely where $Y$ is a random vector taking values in $\mathbb{R}^d$,\\
(2) $Y_n(\omega)$ is a convergent subsequence of $X_n(\omega)$ for almost all $\omega \in \Omega$.
\end{lemma}
\begin{proof}
See e.g. Lemma 2 in \cite{KabStr} or Lemma 1 in \cite{KabRasStr}.
\end{proof}
\begin{remark}
The above claim can be formulated as follows: there exists an increasing sequence of integer-valued random variables $\sigma_k$ such that $X_{\sigma_k}$ converges a.s.
\end{remark}

\begin{lemma}[Kreps-Yan] \label{KY}
Let $K \supseteq -L_+^1$ be a closed convex cone in $L^1$ such that $K \cap L_+^1 = \lbrace 0 \rbrace$. Then there is a probability $\widetilde{P} \sim P$ with $\frac{d\widetilde{P}}{dP} \in L^\infty$ such that $E_{\tilde{P}} \xi \leq 0$ for all $\xi \in K$.
\end{lemma}
\begin{proof}
See e.g. Lemma 3 in \cite{KabStr} or Theorem 2.1.4 in \cite{KabSaf}.
\end{proof}

\begin{proof}[Proof of Theorem \ref{th1}]
\indent (a) $\Rightarrow$ (b) By Lemma \ref{rem3} $\mathcal{A}_t \cap L_+^0(\mathcal{F}_t) = \lbrace 0 \rbrace$ for any $t=1, \ldots, T$ and using Lemma \ref{lem_A_F} also $F_t \cap L_+^0(\mathcal{F}_t) = \lbrace 0 \rbrace$ for any $t=1, \ldots, T$.\\
\indent (b) $\Rightarrow$ (c). Trivial. (Notice that also the implication (a) $\Rightarrow$ (c) is obvious so we could skip the condition (b), which actually we put here to do our analysis more comprehensive.)\\
\indent (c) $\Rightarrow$ (d) To prove this implication we will use the similar technique as in \cite{KabStr} and especially \cite{RygStet} (see Theorem $2.33$). Take any $t$, $k$ such $1 \leq t \leq t+k \leq T$. First notice that by Lemma \ref{lem_Ftk} we have $F_{t-1,t+k} \cap L_+^0(\mathcal{F}_{t+k}) = \lbrace 0 \rbrace$. We will show that the set $F_{t-1,t+k}$ is closed in topology generated by the convergence in pro\-ba\-bi\-li\-ty of measure $\mathbb{P}$. Take a sequence $\xi^n \in F_{t-1,t+k}$ such that $\xi^n \rightarrow \zeta$ in pro\-ba\-bi\-li\-ty. It suffices to show that $\zeta \in F_{t-1,t+k}$. The sequence $\xi^n$ contains a subsequence convergent to $\zeta$ a.s. Thus, at most restricting to this subsequence we can assume that $\xi^n \rightarrow \zeta$, $\mathbb{P}$-a.s. By Lemma \ref{lem_Ftk} for any $n$ there exists $H_t^n \in L^0(\mathbb{R}^d,\mathcal{F}_{t-1})$ and $r_n \in L_+^0(\mathcal{F}_{t+k})$ such that
$$\xi^n = - (H_t^n)^+ \cdot \overline{S}_{t-1} + (H_t^n)^- \cdot \underline{S}_{t-1} + (H_t^n)^+ \cdot \underline{S}_{t+k} - (H_t^n)^- \cdot \overline{S}_{t+k} - r_n \in F_{t-1,t+k},$$
what simply means that $x_{t-1,t+k}(H_t^n) \rightarrow \zeta$, $\mathbb{P}$-a.s.\\
\indent Consider first the situation on the set $\Omega_1 := \lbrace \liminf \Vert H_t^n \Vert < \infty \rbrace \in \mathcal{F}_{t-1}$. By Lemma \ref{lem1} there exists an increasing sequence of integer-valued $\mathcal{F}_{t-1}$-measurable stopping times $\tau_n$ such that $H_t^{\tau_n}$ is convergent a.s. on $\Omega_1$ and for almost all $\omega \in \Omega_1$ the sequence $H_t^{\tau_n(\omega)}(\omega)$ is a convergent subsequence of the sequence $H_t^n(\omega)$. Notice that $H_t^{\tau_n} \in L^0(\mathbb{R}^d,\mathcal{F}_{t-1})$ and $r_{\tau_n} \in L_+^0(\mathcal{F}_{t+k})$ respectively. Let $\tilde{H}_t := \lim\limits_{n \to \infty}H_t^{\tau_n}$. Since $H_t^{\tau_n}$ is convergent, then also $(H_t^{\tau_n})^+$ and $(H_t^{\tau_n})^-$ are convergent. Moreover $(H_t^{\tau_n})^+ \rightarrow (\tilde{H}_t)^+$ and $(H_t^{\tau_n})^- \rightarrow (\tilde{H}_t)^-$. Hence also $r_{\tau_n}$ is convergent a.s. on $\Omega_1$. Define $\tilde{r} := \lim\limits_{n \to \infty}r_{\tau_n}$.
Then
$$\zeta = \lim\limits_{n \to \infty}(- (H_t^n)^+ \cdot \overline{S}_{t-1} + (H_t^n)^- \cdot \underline{S}_{t-1} + (H_t^n)^+ \cdot \underline{S}_{t+k} - (H_t^n)^- \cdot \overline{S}_{t+k} - r_n) =$$
$$\lim\limits_{n \to \infty}(- (H_t^{\tau_n})^+ \cdot \overline{S}_{t-1} + (H_t^{\tau_n})^- \cdot \underline{S}_{t-1} + (H_t^{\tau_n})^+ \cdot \underline{S}_{t+k} - (H_t^{\tau_n})^- \cdot \overline{S}_{t+k} - r_{\tau_n})$$
where the above limit is equal to
$$- (\tilde{H}_t)^+ \cdot \overline{S}_{t-1} + (\tilde{H}_t)^- \cdot \underline{S}_{t-1} + (\tilde{H}_t)^+ \cdot \underline{S}_{t+k} - (\tilde{H}_t)^- \cdot \overline{S}_{t+k} - \tilde{r} \in F_{t-1,t+k}.$$
\indent Consider now the situation on the set $\Omega_2 := \lbrace \liminf \Vert H_t^n \Vert = \infty \rbrace \in \mathcal{F}_{t-1}$. Let us define $G_t^n := \frac{H_t^n}{\Vert H_t^n \Vert}$, $h_n := \frac{r_n}{\Vert H_t^n \Vert}$ and notice that $G_t^n \in L^0(\mathbb{R}^d,\mathcal{F}_{t-1})$ and $h_n \in L_+^0(\mathcal{F}_{t+k})$. We get the convergence
$$- (G_t^n)^+ \cdot \overline{S}_{t-1} + (G_t^n)^- \cdot \underline{S}_{t-1} + (G_t^n)^+ \cdot \underline{S}_{t+k} - (G_t^n)^- \cdot \overline{S}_{t+k} - h_n \rightarrow 0.$$
Similarly as on the set $\Omega_1$ by Lemma \ref{lem1} there exists an increasing sequence of integer-valued $\mathcal{F}_{t-1}$-measurable stopping times $\sigma_n$ such that $G_t^{\sigma_n}$ is convergent a.s. on $\Omega_2$ and for almost all $\omega \in \Omega_2$ the sequence $G_t^{\sigma_n(\omega)}(\omega)$ is a convergent subsequence of the sequence $G_t^n(\omega)$. Let $\tilde{G}_t := \lim\limits_{n \to \infty}G_t^{\sigma_n}$. As previous, notice that by the convergence of the sequence $G_t^{\sigma_n}$ also $(G_t^{\sigma_n})^+$ and $(G_t^{\sigma_n})^-$ are convergent. Moreover $(G_t^{\sigma_n})^+ \rightarrow (\tilde{G}_t)^+$ and $(G_t^{\sigma_n})^- \rightarrow (\tilde{G}_t)^-$. Hence also $h_{\sigma_n}$ is convergent a.s. on $\Omega_2$. Define $\tilde{h} := \lim\limits_{n \to \infty}h_{\sigma_n}$. We get the following equality
$$- (\tilde{G}_t)^+ \cdot \overline{S}_{t-1} + (\tilde{G}_t)^- \cdot \underline{S}_{t-1} + (\tilde{G}_t)^+ \cdot \underline{S}_{t+k} - (\tilde{G}_t)^- \cdot \overline{S}_{t+k} = \tilde{h}.$$
By the condition $F_{t-1,t+k} \cap L_+^0(\mathcal{F}_{t+k}) = \lbrace 0 \rbrace$ we have $\tilde{h} = 0$, $\mathbb{P}$-a.e. Therefore we get
$$(\tilde{G}_t)^+ \cdot (\underline{S}_{t+k} - \overline{S}_{t-1}) - (\tilde{G}_t)^- \cdot (\overline{S}_{t+k} - \underline{S}_{t-1}) = 0, \qquad \mathbb{P}\text{-a.e. on} \; \Omega_2.$$
Since $\tilde{G}_t(\omega) \neq 0$ a.e. on $\Omega_2$ (because $G_t^{\sigma_n}(\omega)$ is a convergent subsequence of the sequence $G_t^n(\omega)$ for almost all $\omega \in \Omega_2$ and for almost all $\omega \in \Omega_2$ we have $\Vert G_t^n(\omega) \Vert = 1$) then there exists a partition of $\Omega_2$ into at most $d$ disjoint subsets $\Omega_2^i \in \mathcal{F}_{t-1}$ such that $\tilde{G}_t^i(\omega) \neq 0$ a. s. on $\Omega_2^i$. (Such a partition can be achieved by choosing $\Omega_2^1 := \lbrace \omega \in \Omega_2 \colon \tilde{G}_t^1(\omega) \neq 0 \rbrace$ and then continuing the partition already on the set $\Omega_2 \setminus \Omega_2^1$ choosing $\Omega_2^2 := \lbrace \omega \in \Omega_2 \setminus \Omega_2^1 \colon \tilde{G}_t^2(\omega) \neq 0 \rbrace$ and so on.) Moreover, any non-empty set $\Omega_2^i$ we can divide into at most two disjoint subsets $\Omega_2^{i,+} := \lbrace \tilde{G}_t^i > 0 \rbrace$ and $\Omega_2^{i,-} := \lbrace \tilde{G}_t^i < 0 \rbrace$. Let us define for any non-empty set $\Omega_2^{i,+}$ and $\Omega_2^{i,-}$ the following sequences
\begin{equation}
(\overline{H}_t^n)^p := (H_t^n)^+ - \beta_n (\tilde{G}_t)^+ \; \text{where} \quad \beta_n := \min\limits_{i: (\tilde{G}_t^i)^+ > 0} \frac{(H_t^{ni})^+}{(\tilde{G}_t^i)^+} \; \text{on} \; \Omega_2^{i,+},
\end{equation}
\begin{equation}
(\overline{H}_t^n)^m := (H_t^n)^- - \beta_n (\tilde{G}_t)^- \; \text{where} \quad \beta_n := \min\limits_{i: (\tilde{G}_t^i)^- > 0} \frac{(H_t^{ni})^-}{(\tilde{G}_t^i)^-} \; \text{on} \; \Omega_2^{i,-}.
\end{equation}
Finally we put $\overline{H}_t^n := (\overline{H}_t^n)^p - (\overline{H}_t^n)^m$ (equivalently for any $i = 1,\ldots,d$ we could define $\overline{H}_t^{ni} = (\overline{H}_t^{ni})^p$ on a non-empty set $\Omega_2^{i,+}$ and $\overline{H}_t^{ni} = -(\overline{H}_t^{ni})^m$ on a non-empty set $\Omega_2^{i,-}$). First notice that $\beta_n$ depends on $\omega \in \Omega_2$ but it is a well defined random variable satisfying the inequality $\beta_n \geq 0$. Moreover $(\overline{H}_t^n)^p \geq 0$ and $(\overline{H}_t^n)^m \geq 0$. Indeed, let us consider the situation on any non-empty set $\Omega_2^{i,+}$. For any $j = 1,\ldots,d$ we have
$$(H_t^{nj})^+ - \frac{(H_t^{nj})^+}{(\tilde{G}_t^j)^+} (\tilde{G}_t^j)^+ = 0 \qquad \text{and} \qquad 0 \leq \beta_n \leq \frac{(H_t^{nj})^+}{(\tilde{G}_t^j)^+}.$$
Hence $(\overline{H}_t^n)^p \geq 0$ and there exists at least one coordinate of $(\overline{H}_t^n)^p$ which is now equal to zero. Notice that this coordinate depends on $\omega$. The situation on the set $\Omega_2^{i,-}$ is analogous. In fact $(\overline{H}_t^n)^+ = (\overline{H}_t^n)^p$ and $(\overline{H}_t^n)^- = (\overline{H}_t^n)^m$. Hence we get
$$x_{t-1,t+k}(\overline{H}_t^n) = - (\overline{H}_t^n)^+ \cdot \overline{S}_{t-1} + (\overline{H}_t^n)^- \cdot \underline{S}_{t-1} + (\overline{H}_t^n)^+ \cdot \underline{S}_{t+k} - (\overline{H}_t^n)^- \cdot \overline{S}_{t+k}$$
$$= - [(H_t^n)^+ - \beta_n (\tilde{G}_t)^+] \cdot \overline{S}_{t-1} + [(H_t^n)^- - \beta_n (\tilde{G}_t)^-] \cdot \underline{S}_{t-1} +$$
$$+ [(H_t^n)^+ - \beta_n (\tilde{G}_t)^+] \cdot \underline{S}_{t+k} - [(H_t^n)^- - \beta_n (\tilde{G}_t)^-] \cdot \overline{S}_{t+k}$$
$$= x_{t-1,t+k}(H_t^n) - \beta_n (- \tilde{G}_t^+ \cdot \overline{S}_{t-1} + \tilde{G}_t^- \cdot \underline{S}_{t-1} + \tilde{G}_t^+ \cdot \underline{S}_{t+k} - \tilde{G}_t^- \cdot \overline{S}_{t+k}) =$$
$$= x_{t-1,t+k}(H_t^n).$$
Summing up $x_{t-1,t+k}(\overline{H}_t^n) = x_{t-1,t+k}(H_t^n)$, $\mathbb{P}$-a.e. on $\Omega_2$ and at least one coordinate of $\overline{H}_t^n$ is equal to zero. However, notice that this coordinate of course may differ in dependence of $\omega \in \Omega_2$.
Now we apply our procedure to the sequence $\overline{\xi}^n := x_{t-1,t+k}(\overline{H}_t^n) - \overline{r}_n \rightarrow \zeta$, $\mathbb{P}$-a.s. on $\Omega_2$.
It is noteworthy that our operations do not affect zero coordinates of the sequence $\overline{H}_t^n$. So by iteration, after a finite number of steps, we construct the desired sequence.\\
\indent (d) $\Rightarrow$ (e) Trivial.\\
\indent (e) $\Rightarrow$ (f) To prove this implication we use some techniques from \cite{Rol} combined with the construction of measure by induction similarly as in \cite{RygStet} (see Corollary $2.35$). Notice that for any random variable $\eta$ there exists a probability measure $P^{'} \sim \mathbb{P}$ such that $\frac{dP^{'}}{d\mathbb{P}} \in L^{\infty}$ and $\eta \in L^1(P^{'})$. Property (d) is invariant under an equivalent change of probability. This consideration allows as to assume without loss of generality that all $\underline{S}_t$, $\overline{S}_t$ are integrable. We will use induction on time horizon or equivalently on $k$. First let $k=0$ and fix any $t \in \lbrace 1, \ldots, T \rbrace$. Define the set $\Psi_{t-1,t} := \overline{F}_{t-1,t} \cap L^1(\mathcal{F}_t)$, which is a closed convex cone in $L^1(\mathcal{F}_t)$. Since we have $\Psi_{t-1,t} \cap L_+^1(\mathcal{F}_t) = \lbrace 0 \rbrace$ then by Lemma \ref{KY} there exists a probability measure $\mathbb{Q}^t \sim \mathbb{P}$ on $(\Omega,\mathcal{F}_t)$ such that $\frac{d\mathbb{Q}^t}{d\mathbb{P}} \in L^{\infty}(\mathcal{F}_t)$ and $E_{\mathbb{Q}^t}\xi \leq 0$ for any $\xi \in \Psi_{t-1,t}$. In particular for
\begin{equation} \label{eq_1}
\xi_{t-1,t}^i = -H_t^i \overline{S}_{t-1}^i + H_t^i \underline{S}_t^i,
\end{equation}
\begin{equation} \label{eq_2}
\tilde{\xi}_{t-1,t}^i = H_t^i \underline{S}_{t-1}^i - H_t^i \overline{S}_t^i
\end{equation}
where $H_t = (0, \ldots, 1{\hskip -4.0 pt}\hbox{1}_A,\ldots,0), \; \mathbb{P}$-a.e., $A \in \mathcal{F}_{t-1}$ and the value $1{\hskip -4.0 pt}\hbox{1}_A$ is on $i$-th position. For the case (\ref{eq_1}) it means that at time $t-1$ if the event $A$ holds we buy $i$-th asset at price $\overline{S}_{t-1}^i$ and liquidate the portfolio at time $t$. For the case (\ref{eq_2}) the situation is opposite, i.e. first we short sale $i$-th asset at time $t-1$ and then we buy it at time $t$. Hence we get the inequalities
$$E_{\mathbb{Q}^t}[(\underline{S}_t^i - \overline{S}_{t-1}^i)1{\hskip -4.0 pt}\hbox{1}_A] \leq 0,$$
$$E_{\mathbb{Q}^t}[(\overline{S}_t^i - \underline{S}_{t-1}^i)1{\hskip -4.0 pt}\hbox{1}_A] \geq 0.$$
Then $E_{\mathbb{Q}^t}(\underline{S}_t^i 1{\hskip -4.0 pt}\hbox{1}_A) \leq E_{\mathbb{Q}^t}(\overline{S}_{t-1}^i 1{\hskip -4.0 pt}\hbox{1}_A)$ and $E_{\mathbb{Q}^t}(\overline{S}_t^i 1{\hskip -4.0 pt}\hbox{1}_A) \geq  E_{\mathbb{Q}^t}(\underline{S}_{t-1}^i 1{\hskip -4.0 pt}\hbox{1}_A)$ for any $i = 1,\ldots,d$ and $A \in \mathcal{F}_{t-1}$. Hence
\begin{equation}
E_{\mathbb{Q}^t}(\underline{S}_t^i \vert \mathcal{F}_{t-1}) \leq E_{\mathbb{Q}^t}(\overline{S}_{t-1}^i \vert \mathcal{F}_{t-1}) = \overline{S}_{t-1}^i,
\end{equation}
\begin{equation}
E_{\mathbb{Q}^t}(\overline{S}_t^i \vert \mathcal{F}_{t-1}) \geq E_{\mathbb{Q}^t}(\underline{S}_{t-1}^i \vert \mathcal{F}_{t-1}) = \underline{S}_{t-1}^i.
\end{equation}
In conclusion, there exists (EBAMM) for the bid-ask process $(\underline{S},\overline{S})$ where $\underline{S} = (\underline{S}_j)_{j=t-1}^t$ and $\overline{S} = (\overline{S}_j)_{j=t-1}^t$.\\
\indent Assume now that the claim is true in a model with the time horizon $k$ where $k \geq 1$. We will show that it is true in a model with the time horizon $k+1$. Fix any $t$, $k$ such that $0 \leq t \leq t+k \leq T$. We show that there exists an equivalent bid-ask martingale measure in the market with the bid-ask process $(\underline{S},\overline{S})$ where $\underline{S} = (\underline{S}_j)_{j=t-1}^{t+k}$ and $\overline{S} = (\overline{S}_j)_{j=t-1}^{t+k}$. By the induction hypothesis there exists (EBAMM) $\mathbb{Q}^{t+k}$ in the market with the bid-ask process $((\underline{S}_j)_{j=t}^{t+k},(\overline{S}_j)_{j=t}^{t+k})$. Notice that the condition (d) is invariant under an equivalent change of probability. Hence we can apply the same method as in the previous part to the probability space $(\Omega,\mathcal{F}_t,\mathbb{Q}^{t+k}_{\vert \mathcal{F}_t})$ where $\mathbb{Q}^{t+k}_{\vert \mathcal{F}_t}$ denotes the measure $\mathbb{Q}^{t+k}$ with the restriction to the $\sigma$-algebra $\mathcal{F}_t$. Then there exists a probability measure $\mathbb{Q}^t \sim \mathbb{Q}^{t+k}_{\vert \mathcal{F}_t}$ such that $\frac{d\mathbb{Q}^t}{d\mathbb{Q}^{t+k}_{\vert \mathcal{F}_t}} \in L^{\infty}$ and the following inequalities are satisfied
$$E_{\mathbb{Q}^t}(\underline{S}_t \vert \mathcal{F}_{t-1}) \leq \overline{S}_{t-1} \qquad \text{and} \qquad E_{\mathbb{Q}^t}(\overline{S}_t \vert \mathcal{F}_{t-1}) \geq \underline{S}_{t-1}.$$
Let us define a probability measure $\mathbb{Q}$ on a measure space $(\Omega,\mathcal{F}_{t+k})$ as follows
\begin{equation}
\frac{d\mathbb{Q}}{d\mathbb{P}} := \frac{d\mathbb{Q}^t}{d\mathbb{Q}^{t+k}_{\vert \mathcal{F}_t}} \frac{d\mathbb{Q}^{t+k}}{d\mathbb{P}}.
\end{equation}
Notice that the density $\frac{d\mathbb{Q}^t}{d\mathbb{Q}^{t+k}_{\vert \mathcal{F}_t}}$ is bounded and $\mathcal{F}_{t}$-measurable hence for any $j \in \lbrace t+1, \ldots, t+k \rbrace$ we have
$$E_{\mathbb{Q}}(\underline{S}_j \vert \mathcal{F}_{j-1}) =
\frac{E_{\mathbb{P}}(\frac{d\mathbb{Q}^t}{d\mathbb{Q}^{t+k}_{\vert \mathcal{F}_t}} \frac{d\mathbb{Q}^{t+k}}{d\mathbb{P}} \underline{S}_j \vert \mathcal{F}_{j-1})}{E_{\mathbb{P}}(\frac{d\mathbb{Q}^t}{d\mathbb{Q}^{t+k}_{\vert \mathcal{F}_{j-1}}} \frac{d\mathbb{Q}^{t+k}}{d\mathbb{P}} \vert \mathcal{F}_{j-1})} =
\frac{E_{\mathbb{P}}(\frac{d\mathbb{Q}^{t+k}}{d\mathbb{P}} \underline{S}_j \vert \mathcal{F}_{j-1})}{E_{\mathbb{P}}(\frac{d\mathbb{Q}^{t+k}}{d\mathbb{P}} \vert \mathcal{F}_{j-1})} =$$
$$= E_{\mathbb{Q}^{t+k}}(\underline{S}_j \vert \mathcal{F}_{j-1}) \leq \overline{S}_{j-1}$$
and on the other hand
$$E_{\mathbb{Q}}(\overline{S}_j \vert \mathcal{F}_{j-1}) =
\frac{E_{\mathbb{P}}(\frac{d\mathbb{Q}^t}{d\mathbb{Q}^{t+k}_{\vert \mathcal{F}_t}} \frac{d\mathbb{Q}^{t+k}}{d\mathbb{P}} \overline{S}_j \vert \mathcal{F}_{j-1})}{E_{\mathbb{P}}(\frac{d\mathbb{Q}^t}{d\mathbb{Q}^{t+k}_{\vert \mathcal{F}_{j-1}}} \frac{d\mathbb{Q}^{t+k}}{d\mathbb{P}} \vert \mathcal{F}_{j-1})} =
\frac{E_{\mathbb{P}}(\frac{d\mathbb{Q}^{t+k}}{d\mathbb{P}} \overline{S}_j \vert \mathcal{F}_{j-1})}{E_{\mathbb{P}}(\frac{d\mathbb{Q}^{t+k}}{d\mathbb{P}} \vert \mathcal{F}_{j-1})} =$$
$$= E_{\mathbb{Q}^{t+k}}(\overline{S}_j \vert \mathcal{F}_{j-1}) \geq \underline{S}_{j-1}.$$
Moreover, it is noteworthy that $E_{\mathbb{Q}}(\overline{S}_t \vert \mathcal{F}_{t-1}) = E_{\mathbb{Q}^t}(\overline{S}_t \vert \mathcal{F}_{t-1})$ and analogously $E_{\mathbb{Q}}(\underline{S}_t \vert \mathcal{F}_{t-1}) = E_{\mathbb{Q}^t}(\underline{S}_t \vert \mathcal{F}_{t-1})$. By induction we conclude that there exists an equivalent bid-ask martingale measure for the bid-ask process $((\underline{S}_t)_{t=0}^T,(\overline{S}_t)_{t=0}^T)$.\\
\indent (f) $\Rightarrow$ (g) As in the previous implication we use induction on time horizon or equivalently on $k$. First let $k=0$ and fix any $t \in \lbrace 1, \ldots, T \rbrace$. Let us define the processes $\hat{S} = (\hat{S}_j)_{j=t-1}^t$ and $\check{S} = (\check{S}_j)_{j=t-1}^t$ as follows
\begin{equation} \label{hat}
\hat{S}_t := \underline{S}_t, \qquad \hat{S}_{t-1} := \max \lbrace \underline{S}_{t-1}, E_{\mathbb{Q}}(\hat{S}_t \vert \mathcal{F}_{t-1}) \rbrace,
\end{equation}
\begin{equation} \label{check}
\check{S}_t := \overline{S}_t, \qquad \check{S}_{t-1} := \min \lbrace \overline{S}_{t-1}, E_{\mathbb{Q}}(\check{S}_t \vert \mathcal{F}_{t-1}) \rbrace.
\end{equation}
Notice that $(\hat{S},\mathbb{Q})$ is (supCPS) and $(\check{S},\mathbb{Q})$ is (subCPS).\\
\indent Assume now that the claim is true in a model with the time horizon $k$ where $k \geq 1$. We will show that it is true in a model with the time horizon $k+1$. Fix any $t$, $k$ such that $0 \leq t \leq t+k \leq T$. We show that there exists (supCPS) $(\hat{S},\mathbb{Q})$ and (subCPS) $(\check{S},\mathbb{Q})$ in the market with the bid-ask process $(\underline{S},\overline{S})$ where $\underline{S} = (\underline{S}_j)_{j=t-1}^{t+k}$ and $\overline{S} = (\overline{S}_j)_{j=t-1}^{t+k}$. By the induction hypothesis there exists (supCPS) $((\hat{S}_j)_{j=t}^{t+k}, \mathbb{Q}^{t+k})$ and (subCPS) $((\check{S}_j)_{j=t}^{t+k}, \mathbb{Q}^{t+k})$ in the market with the bid-ask process $((\underline{S}_j)_{j=t}^{t+k},(\overline{S}_j)_{j=t}^{t+k})$. Notice that the condition (e) is invariant under an equivalent change of probability. Hence we can apply the same method as in the previous part to the pro\-ba\-bi\-li\-ty space $(\Omega,\mathcal{F}_t,\mathbb{Q}^{t+k}_{\vert \mathcal{F}_t})$ where $\mathbb{Q}^{t+k}_{\vert \mathcal{F}_t}$ denotes the measure $\mathbb{Q}^{t+k}$ with the restriction to the $\sigma$-algebra $\mathcal{F}_t$. Then there exists a pro\-ba\-bi\-li\-ty measure $\mathbb{Q}^t \sim \mathbb{Q}^{t+k}_{\vert \mathcal{F}_t}$ such that $\frac{d\mathbb{Q}^t}{d\mathbb{Q}^{t+k}_{\vert \mathcal{F}_t}} \in L^{\infty}$ and the processes $(\hat{S}_j)_{j=t-1}^t$, $(\check{S}_j)_{j=t-1}^t$ defined as in (\ref{hat}), (\ref{check}) are (supCPS), (subCPS). Let us define the stopping time $\tau := \min \lbrace j \geq t-1 \: \vert \: \check{S}_j = \check{S}_t \rbrace$. Then using the optimal stopping theorem the process $\check{S}^{\tau} := (\check{S}_{j \wedge \tau})_{j=t-1}^t$ is a $\mathbb{Q}^t$-martingale. We now define the probability measure $\mathbb{Q}$ on a measure space $(\Omega,\mathcal{F}_{t+k})$ as follows
\begin{equation}
\frac{d\mathbb{Q}}{d\mathbb{P}} := \frac{d\mathbb{Q}^t}{d\mathbb{Q}^{t+k}_{\vert \mathcal{F}_t}} \frac{d\mathbb{Q}^{t+k}}{d\mathbb{P}}.
\end{equation}
Furthermore let $\hat{S}^{'} = (\hat{S}^{'}_j)_{j=t-1}^{t+k}$ be the process of the form
\begin{equation}
\hat{S}^{'}_j = \hat{S}_j \; \text{for any} \; j > t \qquad \text{and} \qquad \hat{S}^{'}_j = \check{S}_{j \wedge \tau} \; \text{for} \; j = t-1, t.
\end{equation}
Notice that $\hat{S}^{'}_{t-1} = \check{S}_{t-1 \wedge \tau} = \check{S}_{t-1}$ and $\hat{S}^{'}_{t} = \check{S}_{t \wedge \tau} = \check{S}_{\tau}$. Hence the inequalities $\underline{S}_j \leq \hat{S}^{'}_j \leq \overline{S}_j$ are also satisfied for $j=t-1,t$. Moreover, the process $\hat{S}^{'} = (\hat{S}^{'}_j)_{j=t-1}^{t+k}$ is a $\mathbb{Q}$-supermartingale. In an analogous way we can construct a $\mathbb{Q}$-submartingale.\\
\indent (g) $\Rightarrow$ (a) It follows from Theorem \ref{th2}.
\end{proof}

\begin{remark}
The conditions from Theorem \ref{th1} are also equivalent to the another one, i.e. for any $1 \leq t \leq t+k \leq T$ there exists an equivalent bid-ask martingale measure $\mathbb{Q}_{t-1}^{t+k}$ for the bid-ask process $\lbrace (\underline{S}_{t-1},\overline{S}_{t-1}), (\underline{S}_{t+k},\overline{S}_{t+k}) \rbrace$ such that $\frac{d\mathbb{Q}_{t-1}^{t+k}}{d\mathbb{P}} \in L^{\infty}$.
\end{remark}

\begin{remark}
Let us define for any $t$, $k$ such that $1 \leq t \leq t+k \leq T$ the following sets $\mathcal{A}_{t-1,t+k} := \mathcal{R}_{t-1,t+k} - L_+^0(\mathcal{F}_{t+k})$ where $\mathcal{R}_{t-1,t+k} := \lbrace x_{t-1,t+k}(H_t) \; \vert \; H_t \in L^0(\mathbb{R}^d,\mathcal{F}_{t-1}) \rbrace$. Then under the assumption $\mathcal{A}_{t-1,t+k} \cap L_+^0(\mathcal{F}_{t+k}) = \lbrace 0 \rbrace$ the set $\mathcal{A}_{t-1,t+k}$ is closed in probability. It suffices to use the analogous reasoning as in the proof of the implication (c) $\Rightarrow$ (d) of Theorem \ref{th1}.
\end{remark}

\begin{corollary} \label{equiv_T=1}
If the time horizon $T=1$ then we have the following equivalence
$$\text{(NA)} \Leftrightarrow \text{(EBAMM)} \Leftrightarrow \text{(CPS)}.$$
\end{corollary}

\begin{remark}
The condition (g) especially says that there exists supCPS $(\hat{S},\mathbb{Q})$. Let us define the Snell envelope of the process $\hat{S}$, i.e.
\begin{equation}
\tilde{S}_T := \hat{S}_T, \qquad \tilde{S}_{t-1} := \max \lbrace \hat{S}_{t-1}, E_{\mathbb{Q}}(\tilde{S}_t \vert \mathcal{F}_{t-1}) \rbrace
\end{equation}
for any $t = 1, \ldots, T$. Notice that by the optimal stopping theorem the random variable $\tau := \min \lbrace t \geq 0 \: \vert \: \tilde{S}_t = \hat{S}_t \rbrace$ is an optimal stopping time and the process $\tilde{S}^{\tau} := (\tilde{S}_{t \wedge \tau})_{t=0}^T$ is a $\mathbb{Q}$-martingale. On the other hand we cannot say that a pair $(\tilde{S}^{\tau},\mathbb{Q})$ is a consistent price system because we do not know whether $\underline{S}_t \leq \tilde{S}_{t \wedge \tau} \leq \overline{S}_t$ or not. We only know that $\underline{S}_{t \wedge \tau} \leq \tilde{S}_{t \wedge \tau} \leq \overline{S}_{t \wedge \tau}$.
\end{remark}

In \cite{GuaLepRas} the following result can be found (see Lemma $6.3$).
\begin{lemma}[Guasoni, L\'{e}pinette, R\'{a}sonyi]
Let $(X_t)_{t \in [0,T]}$ and $(Y_t)_{t \in [0,T]}$ be two c\`{a}dl\`{a}g bounded processes. The following conditions are equivalent:\\
(i) there exists a c\`{a}dl\`{a}g bounded martingale $(M_t)_{t \in [0,T]}$ such that $X \leq M \leq Y$a.s.\\
(ii) for all stopping times $\sigma$, $\tau$ such that $0 \leq \sigma \leq \tau \leq T$ a.s., we have
\begin{equation*}
E[X_{\tau} \vert \mathcal{F}_{\sigma}] \leq Y_{\sigma} \quad \text{and} \quad E[Y_{\tau} \vert \mathcal{F}_{\sigma}] \geq X_{\sigma} \quad \text{a.s.}
\end{equation*}
\end{lemma}
\begin{remark}
However, we cannot use this result in order to strengthen Theorem \ref{th1}, the question of the equivalence between the existence of (EBAMM) and (CPS) in our model is not solved. It is not clear whether Corollary \ref{equiv_T=1} can be extended to the case of any time horizon $T$. This problem also comes down to the following one. Assume that in the model there exists (supCPS) $(\hat{S},\mathbb{Q})$ and (subCPS) $(\check{S},\mathbb{Q})$ under the same probability measure. Is there a consistent price system in this model or not? 
\end{remark}

\begin{remark}
If we assume that there exists the process $S=(S_t)_{t=0}^T$ such that $\overline{S}_t^i = (1 + \lambda^i)S_t^i$ and $\underline{S}_t^i = (1 - \mu^i)S_t^i$
where $0 < \lambda^i, \mu^i < 1$ then our model comes down to a model with proportional transaction costs. Indeed, the value process is then of the form $x_t(H) =$
$$(H \cdot S)_t - \sum_{j=1}^t \lambda (\Delta H_j)^+ S_{j-1} - \sum_{j=1}^t \mu (\Delta H_j)^- S_{j-1} - \lambda (H_t)^- S_t - \mu (H_t)^+ S_t.$$
\end{remark}

\section{The Cox-Ross-Rubinstein model with bid-ask spreads}
\label{CRR}

The model of Cox, Ross and Rubinstein introduced in \cite{CoxRossRub} is very popular in the case of markets without friction and using this model we can estimate the famous Black-Scholes formula. Let us consider this model in the case of markets with bid-ask spreads introduced at the beginning. Define the dynamics of bid and ask processes $\underline{S}$, $\overline{S}$ as follows
\begin{equation}
\underline{S}_t = (1 + \underline{\zeta}_t) \overline{S}_{t-1} \qquad \text{and} \qquad \overline{S}_t = (1 + \overline{\zeta}_t) \underline{S}_{t-1}
\end{equation}
where $(\underline{\zeta}_t)_{t=1}^T$, $(\overline{\zeta}_t)_{t=1}^T$ are sequences of independent and identically distributed random variables. We assume that
$\mathbb{P}(\underline{\zeta}_t = \underline{u}) > 0$, $\mathbb{P}(\underline{\zeta}_t = \underline{d}) = 1 - \mathbb{P}(\underline{\zeta}_t = \underline{u}) > 0$ and in the second case $\mathbb{P}(\overline{\zeta}_t = \overline{u}) > 0$, $\mathbb{P}(\overline{\zeta}_t = \overline{d}) = 1 - \mathbb{P}(\overline{\zeta}_t = \overline{u}) > 0$. To do our analysis much more clear let $p := \mathbb{P}(\underline{\zeta}_t = \underline{u}) = \mathbb{P}(\overline{\zeta}_t = \overline{u})$ and without loss of generality we put $\underline{d} < \underline{u}$, $\overline{d} < \overline{u}$. Because bid and ask processes $\underline{S}$, $\overline{S}$ are strictly positive we get the inequalities $-1 < \underline{d}$ and $-1 < \overline{d}$. For notational simplicity we consider only the one step model. Our purpose is to estimate an equivalent bid-ask martingale measure. First let us define
\begin{equation}
\underline{S}_t^u := \overline{S}_{t-1}(1 + \underline{u}), \qquad \overline{S}_t^u := \underline{S}_{t-1}(1 + \overline{u}),
\end{equation}
\begin{equation}
\underline{S}_t^d := \overline{S}_{t-1}(1 + \underline{d}), \qquad \overline{S}_t^d := \underline{S}_{t-1}(1 + \overline{d}).
\end{equation}
It is illustrated in the figure below.

\begin{figure}[h]
\begin{center}
\includegraphics[width=5cm]{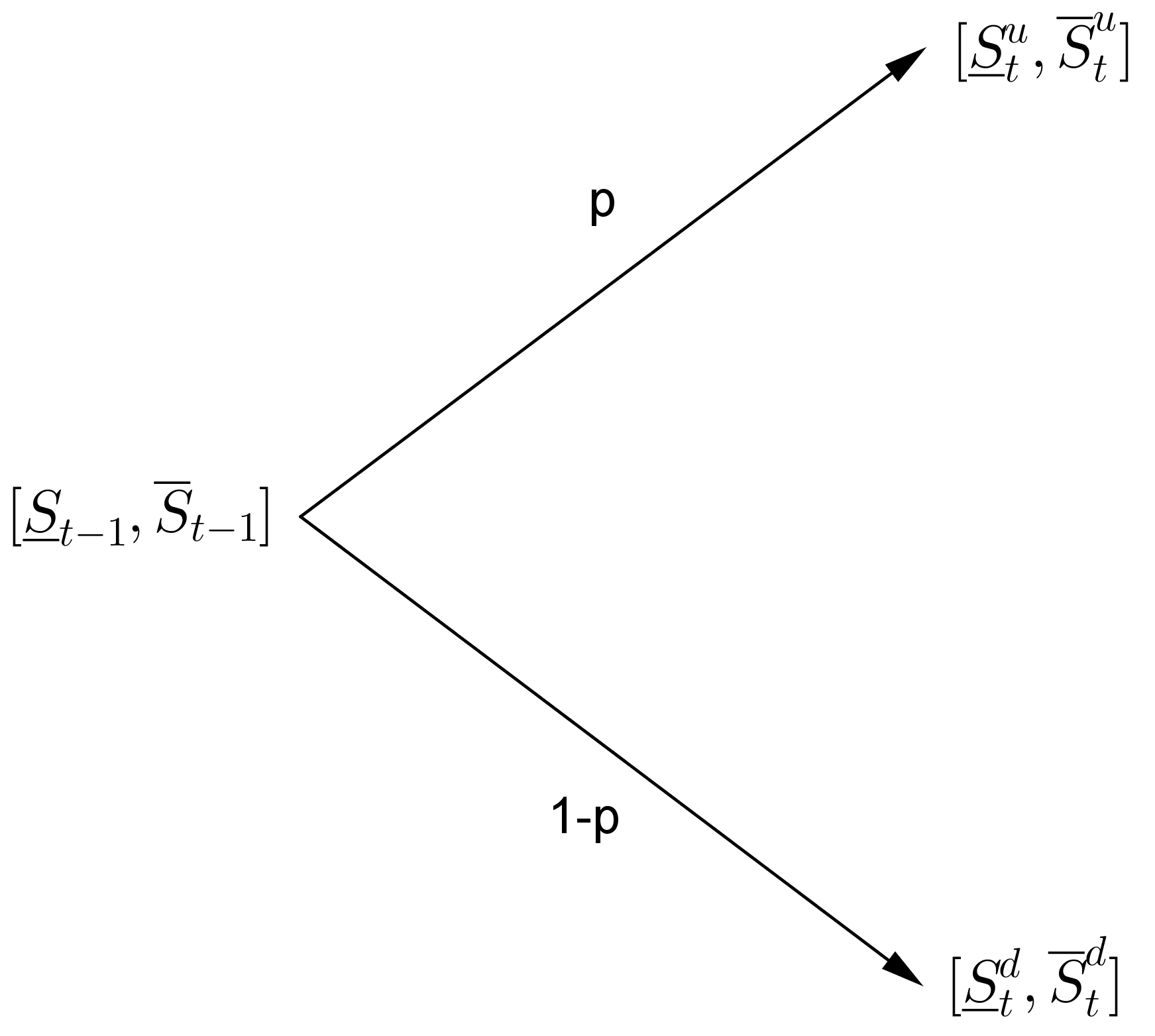}
\caption{}
\end{center}
\end{figure}

This model can be seen as a generalisation of the Cox-Ross-Rubinstein model to the case of bid-ask spreads. Notice that in our model we should also know that $\underline{S}_t \leq \overline{S}_t$. In order to assure this we assume that $\frac{\overline{S}_{t-1}}{\underline{S}_{t-1}} \leq \frac{1 + \overline{u}}{1 + \underline{u}}$ and $\frac{\overline{S}_{t-1}}{\underline{S}_{t-1}} \leq \frac{1 + \overline{d}}{1 + \underline{d}}$. Let us denote the price spread $\Delta := \overline{S}_{t-1} - \underline{S}_{t-1}$. Then the following inequalities should be satisfied
\begin{equation}
1 + \Delta \leq \frac{1 + \overline{d}}{1 + \underline{d}} \qquad \text{and} \qquad 1 + \Delta \leq \frac{1 + \overline{u}}{1 + \underline{u}}.
\end{equation}

By definition an equivalent bid-ask martingale measure (we shall denote it by $q$) should satisfy the following inequalities:
\begin{equation} \label{for_1}
\overline{S}_{t-1} (1 + \underline{u}) q + \overline{S}_{t-1} (1 + \underline{d}) (1 - q) \leq \overline{S}_{t-1},
\end{equation}

\begin{equation} \label{for_2}
\underline{S}_{t-1} (1 + \overline{u}) q + \underline{S}_{t-1} (1 + \overline{d}) (1 - q) \geq \underline{S}_{t-1}.
\end{equation}

Hence from the inequalities (\ref{for_1}) and (\ref{for_2}) we can estimate (EBAMM) as follows:

\begin{equation} \label{est_EBAMM}
\frac{- \overline{d}}{\overline{u} - \overline{d}} \leq q \leq \frac{- \underline{d}}{\underline{u} - \underline{d}}.
\end{equation}

Notice that for the existence of (EBAMM) we need to know that the following inequalities are satisfied.

\begin{equation}
\frac{- \overline{d}}{\overline{u} - \overline{d}} \leq \frac{- \underline{d}}{\underline{u} - \underline{d}}
\end{equation}
\begin{equation}
0 < \frac{- \underline{d}}{\underline{u} - \underline{d}} \qquad \text{and} \qquad \frac{- \overline{d}}{\overline{u} - \overline{d}} < 1.
\end{equation}
These conditions assure that there exists at least one $q \in (0,1)$. Then we get the necessary and sufficient conditions for the existence of (EBAMM), i.e.

\begin{equation} \label{eq_cond_CRR}
\underline{d} < 0 < \overline{u} \qquad \text{and} \qquad \underline{d} \: \overline{u} \leq \overline{d} \: \underline{u}.
\end{equation}

It is noteworthy that by Theorem \ref{th1} these conditions are equivalent to the absence of arbitrage in the Cox-Ross-Rubinstein model with bid-ask spreads. Now we illustrate our model by some examples.

\begin{example}
Consider the model with one risky asset and the time horizon $T=1$ presented in the Figure $2$. In this model there is no arbitrage and $q \in [0,1]$. Hence we can take any $q \in (0,1)$ to get an equivalent bid-ask martingale measure. Furthermore notice that $\underline{d} = - \frac{3}{4}$, $\overline{d} = \underline{u} = 0$, $\overline{u} = 3$ and the conditions from (\ref{eq_cond_CRR}) are satisfied. Indeed, we have $\underline{d} = -\frac{3}{4} < 0 < 3 = \overline{u} \:$ and $\: \underline{d} \: \overline{u} = - \frac{9}{4} \leq 0 = \overline{d} \: \underline{u}$.
\end{example}

\begin{figure}[h]
\begin{center}
\includegraphics[width=5cm]{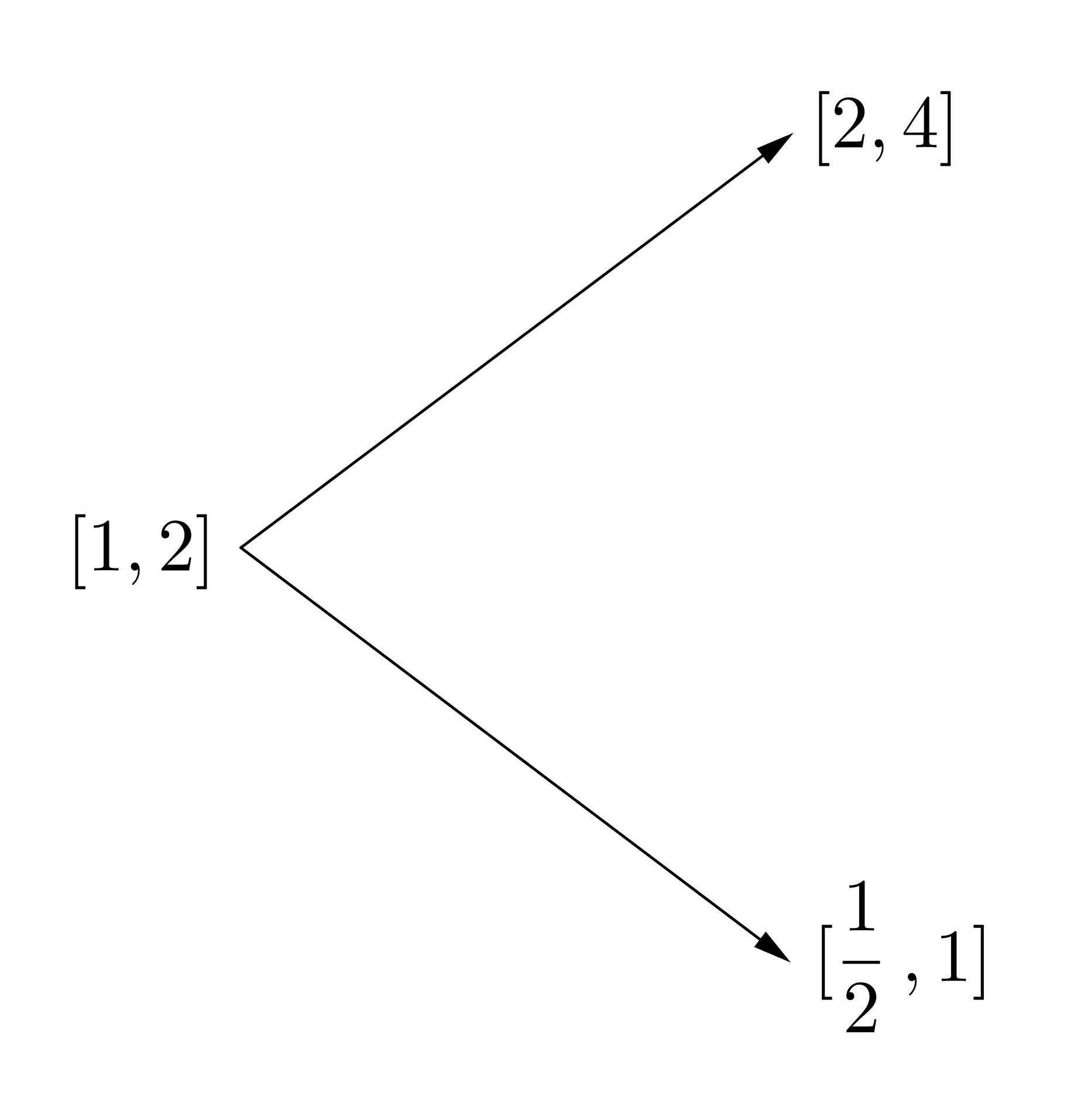}
\caption{}
\end{center}
\end{figure}
 
\begin{example}
Consider now the model with the time horizon $T=1$ and one risky asset presented in the Figure $3$.
\begin{figure}
\begin{center}
\includegraphics[width=5cm]{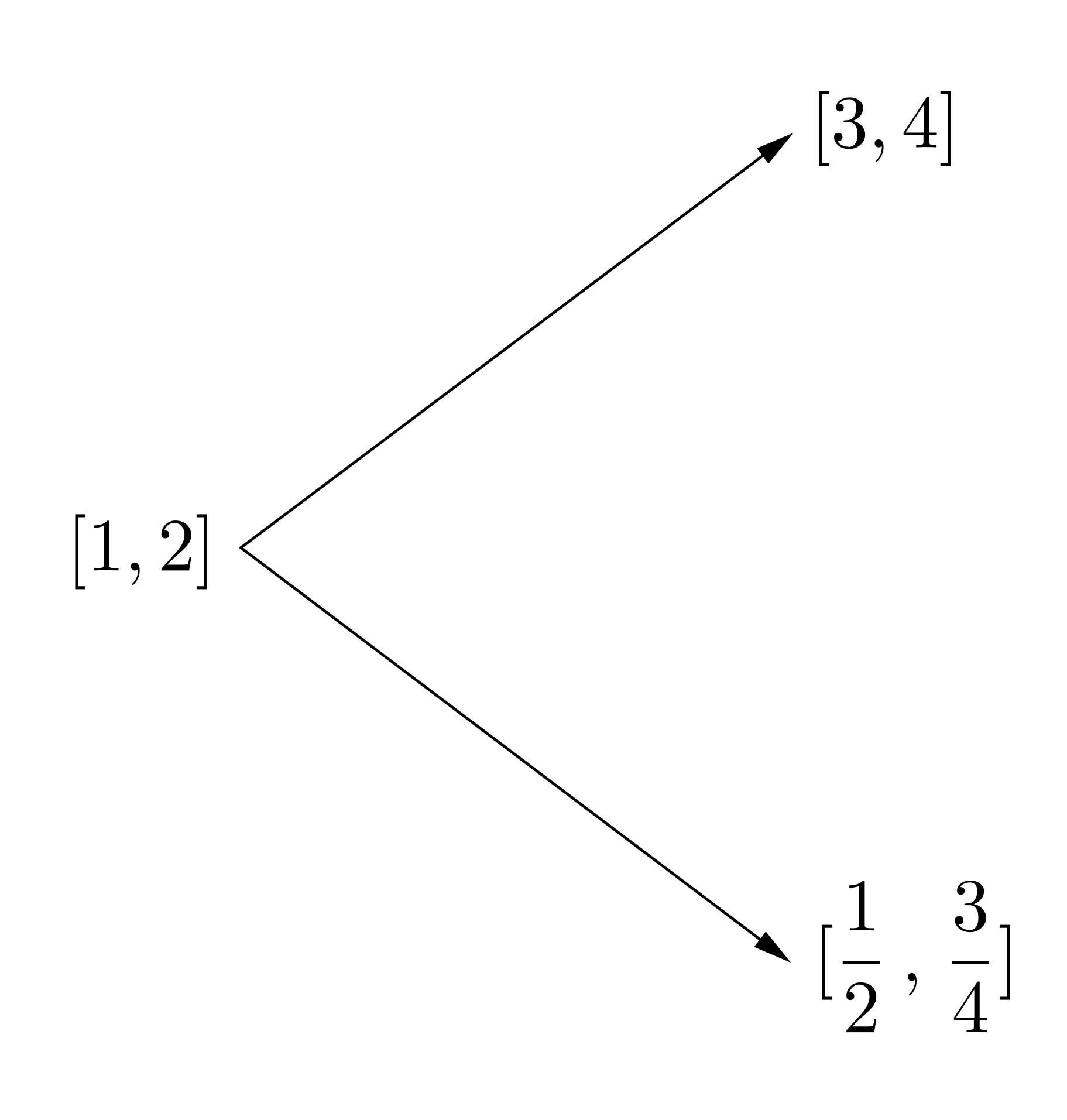}
\caption{}
\end{center}
\end{figure}
In this model we have $\overline{u} = 3$, $\overline{d} = - \frac{1}{4}$, $\underline{u} = \frac{1}{2}$, $\underline{d} = - \frac{3}{4}$ and by (\ref{est_EBAMM}) $q \in [\frac{1}{13},\frac{3}{5}]$. Notice that also the conditions from (\ref{eq_cond_CRR}) are satisfied. Indeed, we have $\underline{d} = -\frac{3}{4} < 0 < 3 = \overline{u}$ and $\underline{d} \: \overline{u} = - \frac{9}{4} \leq  - \frac{1}{8} = \overline{d} \: \underline{u}$.
\end{example}

\subsection*{Acknowledgements}
The author would like to thank Professor {\L}ukasz Stettner for numerous helpful suggestions and comments, which improved this paper.

$$ $$

\end{document}